\documentclass[12pt,draftcls,journal,onecolumn]{IEEEtran}
\usepackage{amsfonts,color,morefloats,pslatex,graphicx,graphics}
\usepackage{amssymb,amsthm, amsmath,latexsym}

\newtheorem{theorem}{Theorem}
\newtheorem{lemma}[theorem]{Lemma}
\newtheorem{remark}[theorem]{Remark}

\newtheorem{open}[theorem]{Open Problem}

\newtheorem{example}[theorem]{Example}

\newcommand{\ord}{{\mathrm{ord}}}

\newcommand{\gf}{{\mathrm{GF}}}

\newcommand{\C}{{\mathcal{C}}}

\usepackage{blindtext}

\ifCLASSINFOpdf

\else

\fi

\begin{document}
%
\title{Four constructions of self-dual binary cyclic codes with a lower bound on the minimum distances better than the square-root bound
}

\author{Xiaoqiang Wang, Xun Song, Dabin Zheng, Hao Chen, Cunsheng Ding

\thanks{Xiaoqiang Wang, Xun Song and Dabin Zheng are with the Hubei Key Laboratory of Applied Mathematics, Faculty of Mathematics and Statistics, Hubei University, Wuhan, 430062, China (e-mail: waxiqq@163.com; xvnsong@126.com; dzheng@hubu.edu.cn). Hao Chen is with the College of Information Science and Technology, Jinan University, Guangzhou, Guangdong 510632, China (e-mail: haochen@jnu.edu.cn). Cunsheng Ding is with the Department of Computer Science and Engineering, The Hong Kong University of Science and Technology, Hong Kong,
China (e-mail: cding@ust.hk). The corresponding author is Dabin Zheng.
}

}

\maketitle
\hskip 1mm
\begin{abstract}
In spite of the intensive study of cyclic codes and the recent construction of an infinite family of self-dual binary cyclic codes whose minimum distances have the square-root bound in IEEE Trans. IT, vol. 71, no. 4,  2025,
it is still a 70-year-old open problem whether there is an infinite family of self-dual binary cyclic codes whose minimum distances have a lower bound better than the square-root bound.  This paper settles this long-standing open problem in coding theory by presenting infinite families of such self-dual binary cyclic codes.
As by-products,
several families of cyclic codes with better parameters than those in some references are also
constructed in this paper.
\end{abstract}




%

\section{Introduction}\label{sec-intro}

Let $\mathbb{F}_q$ denote the finite field with $q$ elements, where $q$ is a prime power.
Let $n$ and $k$ be positive integers satisfying $k \leq n$.
A linear code $\mathcal{C}$ over $\mathbb{F}_q$ with parameters $[n, k, d]$ is a $k$-dimensional linear subspace of $\mathbb{F}_q^n$ having minimum Hamming distance $d$.
Its dual code, denoted by $\mathcal{C}^{\perp}$, is defined as
\[
\mathcal{C}^{\perp} = \bigl\{ \mathbf{b} \in \mathbb{F}_q^n : \mathbf{b} \mathbf{c}^{T} = 0 \text{ for all } \mathbf{c} \in \mathcal{C} \bigr\},
\]
where $\mathbf{b} \mathbf{c}^{T}$ stands for the standard inner product of the vectors $\mathbf{b}$ and $\mathbf{c}$.
A code $\mathcal{C}$ is called a cyclic code if it is closed under the cyclic shift operation: whenever $(c_0, c_1, \dots, c_{n-1})$ belongs to $\mathcal{C}$, so does $(c_{n-1}, c_0, c_1,
\dots, c_{n-2})$.
By identifying each vector $(c_0, c_1, \dots, c_{n-1}) \in \mathbb{F}_q^n$ with the polynomial
\[
c_0 + c_1 x + c_2 x^2 + \cdots + c_{n-1} x^{n-1} \in \mathbb{F}_q[x]/\langle x^n - 1 \rangle,
\]
any cyclic code $\mathcal{C}$ of length $n$ over $\mathbb{F}_q$ corresponds to a subset of the quotient ring $\mathbb{F}_q[x]/\langle x^n - 1 \rangle$.
Since every ideal of this ring is principal, $\mathcal{C}=\langle g(x)\rangle$ for a unique monic polynomial $g(x)$ of minimal degree, called the \emph{generator polynomial} of $\C$. The \emph{check polynomial} is $h(x)=(x^n-1)/g(x)$, and the roots of $g(x)$ (resp. $h(x)$) are called the \emph{zeros} (resp. \emph{nonzeros}) of $\mathcal{C}$.

Assume $\gcd(n,q)=1$ and let $\ell=\mathrm{ord}_n(q)$ be the multiplicative order of $q$ modulo $n$. Let $\alpha$ be a generator of $\mathbb{F}_{q^\ell}^*$,  then  $\beta=\alpha^{(q^\ell-1)/n}$ is a primitive $n$th root of unity. For $0\le i\le q^\ell-2$, denote by $m_i(x)$ the minimal polynomial of $\beta^i$ over $\mathbb{F}_q$. For two integers $\delta$ with $2\le\delta\le n$ and $b$, the cyclic code with generator polynomial
\[
g_{(\delta,b)}(x)=\operatorname{lcm}\big(m_b(x),m_{b+1}(x),\dots,m_{b+\delta-2}(x)\big)
\]
is called a \emph{BCH code} of length $n$ and designed distance $\delta$. If $b=1$ and $n=q^m-1$, the code is referred to as a \emph{narrow-sense primitive BCH code} and will be denoted simply by $\mathcal{C}_\delta$ in this paper.

BCH codes form a prominent class of cyclic codes and are widely used for error correction in communication and storage systems. Binary BCH codes were introduced around 1960 by Bose, Ray-Chaudhuri, and Hocquenghem \cite{Bose62,Hocquenghem59}, and later extended to arbitrary finite fields by Gorenstein and Zierler \cite{Gorenstein61}. Despite decades of intensive study, the exact parameters of BCH codes and their duals are known only in some special cases. For further information on BCH codes we refer the reader to \cite{Charpin90,Ding15,Ding17,Fan23,Fu24,GDL21,Liu17,Wang23,Yue15} and the references therein.

The construction of self-dual codes has been an active research topic for many years. For finite lengths, numerous constructions of self-dual generalized Reed-Solomon (GRS) codes have been proposed (see, e.g., \cite{Fang19,Grass08,Jin17,Zhang20}). For a binary self-dual code of length $n$, its minimum distance $d$ is bounded by $d\le 4\lfloor n/24\rfloor+6$ if $n\equiv22\pmod{24}$, and $d\le4\lfloor n/24\rfloor+4$ otherwise.
A self-dual binary linear code attaining any of these two bounds is called an \emph{extremal code} \cite[Chapter 19]{MS78}. Ternary self-dual codes satisfy $d\le3\lfloor n/12\rfloor+3$,  and are said to be {\it extremal}  if the equality holds. Constructions of extremal or optimal self-dual codes over small fields, as well as those with large minimum distances,
can be found in \cite{Betsumiya23,Dougherty97,Gaborit03,Gulliver08,Harada07,Shi18}. For the best known minimum distances of self-dual codes over small fields, we refer to \cite{Huffman03}.

A lower bound on the minimum distance of an $[n,k,d]$ code over $\mathbb{F}_q$ is called a \emph{square-root bound} if it is $\sqrt{n}$, and a \emph{square-root-like bound} if it is $c\sqrt{n}$ for a fixed positive constant $c$ depending only on $q$. To date, only a few infinite families of self-dual codes over arbitrary finite fields with unbounded length $n$ and minimum distance $d\ge\sqrt{n}$ are known. To the best of our knowledge, Chen and Ding \cite{Chen23} gave the first infinite family of self-dual binary cyclic codes attaining exactly the square-root bound, and they pointed out the following long-standing open problem:

\begin{open}\cite{Chen23}
Does there exist an infinite family of self-dual binary cyclic codes with a lower bound on their minimum distances better than the square-root lower bound?
\end{open}

In addition to the infinite family of self-dual cyclic codes over $\mathbb{F}_{2^s}$ in \cite{Chen23},
the known infinite families of linear self-dual (non-cyclic) codes with $d\ge\sqrt{n}$ are:
\begin{enumerate}
\item The extended codes of odd-like quadratic residue codes \cite{Huffman03}.

\item The quasi-cyclic binary codes with parameters $[2q + 2,q + 1,d]$, where $q \equiv 3 \pmod{8}$ is a power of
 an odd prime and $(d - 1)^2 - (d - 1) + 1 \geq 2q + 1$ \cite{calderbank1983square}.

\item The Pless symmetry codes \cite{ding2022designs}, \cite{pless1972symmetry}.

\item The generalized Reed-Muller code of order $[m(q - 1) - 1]/2$ over $\mathbb{F}_q$ with parameters $[q^m,q^m/2,[(q + 2)/2]q^{(m-1)/2}]$, where $m$ is odd and $q=2^s$ with a positive
    integer $s$ \cite[Theorem 5.8, Theorem 5.25]{assmus1998polynomial}.
\item A family of negacyclic codes over $\mathbb{F}_q$ with parameters (see \cite{sun2023two})

\[
\left[ \frac{q^m - 1}{2}, \frac{q^m - 1}{4}, d \right],
\]

where $q$ is an odd prime power, $m \geq 2$ is an integer, $q^m \equiv 1 \pmod{4}$ and

\[
d =
\begin{cases}
q^{m/2}, & \text{if } m \text{ is even}, \\
\left( \frac{q+3}{4} \right) q^{(m-1)/2}, & \text{if } m \text{ is odd}.
\end{cases}
\]

\end{enumerate}

The primary motivation of this paper is the absence of an infinite family of binary self-dual cyclic codes whose
minimum distances have a lower bound exceeding the square-root bound.
By exploiting the defining sets of cyclic codes of length $2^m-1$, we develop four construction methods that yield infinite families of self-dual binary cyclic codes of length $2^{m+1}-2$ whose minimum distance have a lower bound larger than the square-root bound,  thereby providing an affirmative answer to the 70-year-old open problem above.

A second motivation arises from a recent paper by Tang and Ding \cite{TangDing22}, where several infinite families of binary cyclic codes with parameters $[n,(n\pm1)/2,d]$ and good minimum distance $d$ were constructed, where $n=2^m-1$. We present new families of cyclic codes that achieve both larger dimensions and better lower bounds on the minimum distances than those in \cite{TangDing22}.

The third motivation comes from known lower bounds on the minimum distance of the duals of the narrow-sense primitive BCH codes over $\mathbb{F}_{2^m}$. Known bounds include the Sidelnikov bound, the Carlitz-Uchiyama bound, and the recent improvements in \cite{GDL21} and \cite{Wang24}. However, for $m\ge9$ and certain ranges of the designed distance, these bounds can still be significantly improved in many cases. In this paper, when the designed distance of the binary narrow-sense primitive BCH codes of length $2^m-1$ falls within certain ranges, we establish new lower bounds that are tighter than the existing ones in many cases.

The remainder of the paper is organized as follows. Section II recalls necessary preliminaries. Section III presents four families of self-dual binary cyclic codes whose minimum distances have a lower bound better than the square-root bound, and shows several classes of binary cyclic codes with  large dimensions and good minimum distances. It also provides improved lower bounds on the minimum distances of the duals of the binary narrow-sense primitive BCH codes for the designed distance in a certain interval. Section IV concludes the paper.

\section{Preliminaries}

In this section, we present some
basic concepts and results that will be utilized subsequently. Unless stated otherwise, from this point forward, we will use the notation outlined below:
\begin{itemize}
\item $\mathbb{F}_{q^m}$ is the finite field with $q^m$ elements, where $q$ is a prime power and $m$ is a positive integer.
\item  $\beta$ is an $n$-th primitive root of unity in $\mathbb{F}_{q^m}$.
\item $\mathcal{C}$ denotes the cyclic code of length $n$ with generator polynomial $g(x)$.
\item $T(\mathcal{C})=\{0\leq i \leq n-1 \,|\, g(\beta^{i})=0 \}$ is the defining set of $\mathcal{C}$ with respect to $\beta$, and $T(\mathcal{C})^{-1}=\{n-i\,|\,i\in
    T(\mathcal{C})\}$.
\item $\mathcal{C}^{\perp}$  is the dual of $\mathcal{C}$ and $T(\mathcal{C}^\perp)$ is the defining set of  $\mathcal{C}^{\perp}$.
\item ${\rm CL}(a)$ denotes the $q$-cyclotomic coset leader modulo $n$ containing $a$, where $a$ is a positive integer with $1\leq a \leq n-1$.
\item $d(\mathcal{C})$ and $d(\mathcal{C}^\perp)$ denote the minimum Hamming distances of $\mathcal{C}$ and $\mathcal{C}^{\perp}$, respectively.
\item $[a,b]$ denotes the integer interval $\{a,a+1,\ldots,b\}$.
\item $\mathbb{Z}_n = \mathbb{Z}/n\mathbb{Z}$  denotes the residue class ring modulo $n$.
\end{itemize}

\subsection{Known results of several families of $[n, (n \pm 1)/2]$ binary cyclic codes }

Let $s$ be an integer with $0\leq s<n$. The \emph{$q$-cyclotomic coset} of $s$ modulo $n$ is defined by
$$ C_{s}=\{s,sq,sq^2,\ldots,sq^{\ell_{s-1}}\}\,\, {\text\,\bmod \,\,n\subseteq \mathbb{Z}_n, }$$
where $i \bmod{n}$ denotes the unique integer $r$ with $0 \leq r \leq n-1$ such that $i \equiv r \pmod{n}$,
 $\ell_s$ is the smallest positive integer such that $s\equiv sq^{\ell_s} \pmod n$, and $\ell_s$ is the size of the $q$-cyclotomic coset.
The smallest integer in $C_{s}$ is called the \emph{coset leader} of $C_{s}$.

We will need the following result.

\begin{lemma}\cite[Chapter 4]{Huffman03}
\label{size about cyclotomic coset}
Let $\gcd(q,n)=1$.
The size of each $q$-cyclotomic coset is a divisor of $\ord_{n}(q)$, where $\ord_n(q)$ is the smallest positive integer $k$ such that $q^k\equiv 1\pmod{n}$.
\end{lemma}

Let $m \geq 2$ be a positive integer and let $n=2^m-1$. Let $i$  be an integer with $0\leq i \leq n$, then the $2$-adic expansion of $i$ can be written as
 $$\bar{i}=(i_{0},i_{1},\ldots,i_{m-1})_2,$$ where
  $i=i_{0}+i_{1}2+\cdots+i_{m-2}2^{m-2}+i_{m-1}2^{m-1}$ and each $i_j\in \{0,1\}$. Define $\omega_2{(\bar{i})}=\sum\limits_{j=0}^{m-1}i_j$.

Let $\alpha$ be a primitive element of $\mathbb{F}_{2^m}$. The authors in \cite{TangDing22} defined a polynomial
\begin{equation*}
g_{(i,m)}=\prod\limits_{1\leq j \leq n-1 \atop \omega_2(\bar{j}) \equiv i \pmod{2} } (x-\alpha^j)
\end{equation*}
over $\mathbb{F}_2$ for each $i \in \{ 0 , 1\}$.  Let $\mathcal{C}_{(i,m)}$ denote the cyclic code over $\mathbb{F}_2$ with length $n$ and generator polynomial $g_{(i,m)}(x)$.  They
showed parameters of these codes as follows.

\begin{lemma}\label{lem:02061}\cite[Theorem 12]{TangDing22}
Let $m\geq 4$ be even. Then $\mathcal{C}_{(0,m)}$ and $\mathcal{C}_{(1,m)}$ have parameters
 $[2^m-1,2^{m-1}+1,d_1]$ and $[2^m-1,2^{m-1}-1,d_2]$, respectively, where
\begin{align*}
d_1\geq \begin{cases}
2^{\frac{m-4}{2}}+1, \; & \text{if $m \equiv 4 \pmod{8}$,}\\
2^{\frac{m-2}{2}}+1, \; & \text{if $m \equiv 2 \pmod{4}$}
\end{cases} \,\,\,\,\,\,\,\,\text{and}\,\,\,\,\,\,\,\,
d_2\geq \begin{cases}
2^{\frac{m-2}{2}}+1, & \text{if $m \equiv 0 \pmod{4} $},\\
2^{\frac{m-4}{2}}+1, & \text{if $m \equiv 2 \pmod{4} $}.
\end{cases}
\end{align*}
\end{lemma}

\begin{lemma}\label{lem:0206}\cite[Theorem 13]{TangDing22}
Let $m\geq 4$ be even. Then $\mathcal{C}_{(0,m)}^\perp$ and $\mathcal{C}_{(1,m)}^\perp$  have parameters
 $[2^m-1,2^{m-1}-2,d_1^\perp]$ and $[2^m-1,2^{m-1},d_2^\perp]$, where
 \begin{align*}
d_1^\perp \geq \begin{cases}
2^{\frac{m-2}{2}}+2, \; & \text{if $m \equiv 0 \pmod{4}$,}\\
2^{\frac{m-4}{2}}+2, \; & \text{if $m  \equiv 2 \pmod{4}$}
\end{cases} \,\,\,\,\,\,\,\,\text{and}\,\,\,\,\,\,\,\,
d_2^\perp\geq \begin{cases}
2^{\frac{m-4}{2}}+2, & \text{if $m \equiv 4 \pmod{8} $},\\
2^{\frac{m-2}{2}}+2, & \text{if $m \equiv 2 \pmod{4} $}.
\end{cases}
\end{align*}
\end{lemma}

\begin{lemma}\label{lem:02062}\cite[Theorem 14]{TangDing22}
 Let $m\geq 3$ be odd. Then $\mathcal{C}_{(0,m)}$ and $\mathcal{C}_{(1,m)}$ have parameters
 $[2^m-1,2^{m-1},d]$, where
\begin{equation*}
\begin{split}
d \geq \begin{cases}
2^{\frac{m-1}{2}}+1, \; & \text{if $m \equiv 3 \pmod{4}$}, \\
2^{\frac{m-1}{2}}+3, \; & \text{if $m  \equiv 1 \pmod{4}$}.
\end{cases}
\end{split}
\end{equation*}
\end{lemma}

\begin{lemma}\label{lem:02063}\cite[Theorem 15]{TangDing22}
 Let $m\geq 3$ be odd. Then $\mathcal{C}_{(0,m)}^\perp$ and $\mathcal{C}_{(1,m)}^\perp$ have parameters
 $[2^m-1,2^{m-1}-1,d^{\perp}]$, where
\begin{equation*}
\begin{split}
d^\perp \geq \begin{cases}
2^{\frac{m-1}{2}}+2, \; & \text{if $m \equiv 3 \pmod{4}$},\\
2^{\frac{m-1}{2}}+4, \; & \text{if $m  \equiv 1 \pmod{4}$}.
\end{cases}
\end{split}
\end{equation*}
\end{lemma}

In this paper, we will present new families of binary cyclic codes that have both larger dimensions and better lower bounds on the minimum distances than the binary cyclic codes $\C_{(i,m)}$ above.  The purpose of
introducing the four lemmas above is for comparison later.

\subsection{Several basic results on linear codes}
Let $\mathcal{C}$ be a linear code of length $n$ over $\mathbb{F}_q$. Then $\mathcal{C}$ is said to be self-orthogonal if $\mathcal{C}\subseteq
\mathcal{C}^{\perp}$,  dual-containing  if $\mathcal{C}^{\perp}\subseteq\mathcal{C}$  and  self-dual  if $\mathcal{C}^{\perp}=\mathcal{C}$. For these special codes, we have the
following related results, which will be used later.

\begin{lemma} \cite[Theorem 2]{grassl1997codes}, \cite[Chapter 4]{Huffman03}
\label{criterion for self-orthogonal}
Let $\mathcal{C}$ be a cyclic code of length $n$ over $\mathbb{F}_{q}$ such that $gcd(n, q)=1$. The code $\mathcal{C}$ is dual-containing if and only if $ T(\mathcal{C}) \cap
T(\mathcal{C})^{-1} = \emptyset$, and $\mathcal{C}$ is self-orthogonal if and only if $T(\mathcal{C}) \cup T(\mathcal{C})^{-1} = \mathbb{Z}_n$.
\end{lemma}

\begin{lemma}\cite[Theorem 2]{Chen23} 
\label{criterion for plotkin-self-dual-1}
Let $s$ be a positive integer. Suppose that $\mathcal{C} \subseteq \mathbb{F}_{2^{s}}^{n}$ is a dual-containing linear code. Then the linear
code
\begin{equation*}
\mathcal{D}_{1} = \{ (\mathbf{u} \mid \mathbf{u} + \mathbf{v}) : \mathbf{u} \in \mathcal{C},\ \mathbf{v} \in \mathcal{C}^{\perp} \} \subseteq \mathbb{F}_{2^{s}}^{2n}
\end{equation*}
is self-dual and has minimum distance $\min\{d(\mathcal{C}^{\perp}), 2d(\mathcal{C})\}$.
\end{lemma}

If $C$ is self-orthogonal, then $\,C^{\bot}\,$ is dual-containing.  Then by Lemma \ref{criterion for plotkin-self-dual-1} we have the following result.
\begin{lemma} 
\label{criterion for plotkin-self-dual-2}
Let $s$ be a positive integer.  If $\mathcal{C} \subseteq \mathbb{F}_{2^{s}}^{n}$ is a self-orthogonal linear code, then the linear code
\begin{equation*}
\mathcal{D}_{2} = \{ (\mathbf{u} \mid \mathbf{u} + \mathbf{v}) : \mathbf{u} \in \mathcal{C}^{\perp},\ \mathbf{v} \in \mathcal{C} \} \subseteq \mathbb{F}_{2^{s}}^{2n}
\end{equation*}
is self-dual and has minimum distance $\min\{d(\mathcal{C}), 2d(\mathcal{C}^{\bot})\}$.
\end{lemma}

The following two results will be used repeatedly in this paper.

\begin{lemma}(The Generalized Hartmann-Tzeng Bound)\,\cite[Theorem 2]{ROOS1982229}
\label{generalizations of bch bound.}
Consider a cyclic code $\mathcal{C}$ of length $n$ with defining set $T(\mathcal{C})$.
 Define the set  $N=\{b+i_1c_1+i_2c_2+\cdots+i_kc_k \,|\, 0 \leq i_j\leq s_j, 1\leq j \leq k \}$, where ${\rm
 gcd}(n,c_1)=1$, and ${\rm gcd}(n,c_j)<2+s_1+s_2+\cdots+s_{j-1}$ if $2\leq j \leq k$. If $N \subseteq T(\mathcal{C}) $,  then $d(\mathcal{C})\geq 2+s_1+\cdots+s_k$.
\end{lemma}

\begin{lemma} \cite[Theorem 1]{Chen23}
\label{lemma:generalized-van-lint} 
Let $q$ be a power of $2$ and $n$ be an odd positive integer. Let $\mathcal{C}_1 \subseteq \mathbb{F}_{q}^n$ be a cyclic code with generator polynomial $g_1(x) \in \mathbb{F}_q[x]$
and ${\mathcal{C}}_2 \subseteq \mathbb{F}_q^n$ be a cyclic code generated by the polynomial $g_1(x)g_2(x) \in \mathbb{F}_q[x]$, where $g_2(x)$ is a divisor of $x^n + 1$. Then the code $\{
(\mathbf{u} \mid \mathbf{u} + \mathbf{v}) : \mathbf{u} \in \mathcal{C}_1,\ \mathbf{v} \in \mathcal{C}_2 \}$ is permutation-equivalent to the repeated-root cyclic code $\C'$ of length \(
2n \) with the generator polynomial  \( g_1(x)^2g_2(x)/\gcd(g_1(x),g_2(x)) \).

\end{lemma}
%
%
%

To construct a cyclic code $\C'$ over $\mathbb{F}_{2^s}$ with good parameters using Lemma \ref{lemma:generalized-van-lint}, one has to choose  the two cyclic codes $\C_1$ and $\C_2$
carefully.

\section{Four constructions of self-dual binary cyclic codes}

Let $\C$ be a binary cyclic code with length $n=2^m-1$ and
 $T(\mathcal{C})$ be the defining set of $\mathcal{C}$ with respect to $\alpha$, where $\alpha$ is a primitive element in $\mathbb{F}_{2^m}$.  Let $T(\C^{\perp})$ denote the defining set of the dual code $\mathcal{C}^{\perp}$ with respect to
 $\alpha$. It is easy to check that
 \begin{equation}\label{eq:0131}
 T(\C^{\perp})=\mathbb{Z}_n \setminus T(\C)^{-1}.
 \end{equation}
 In the following subsections, we will construct several classes of self-dual binary cyclic codes with length $2^{m+1}-2$,  present several classes of cyclic codes
 with good parameters, and develop a lower bound on the minimum distance of the duals of the narrow-sense binary BCH codes with the designed distances in some ranges.
%
%

\subsection{The first construction of self-dual binary cyclic codes}

In this subsection, let $n = 2^m - 1$, where $m$ is an even integer with $m \geq 6$. Let $l$ be a positive integer, $F_i= 2^{ \frac{m}{2}+ f_i} + 2^{\frac{m}{2}} - 1$ and
$A_l=C_{F_1}\cup C_{F_2}\cup \cdots \cup C_{F_l}$, where $1\leq i\leq l$ and $1 \leq f_1\neq f_2\neq\cdots\neq f_l\leq \frac{m}{2}- 2$. Let
$\mathcal{C}$ be the binary cyclic code of length $n$ with the defining set
\begin{equation}\label{0318}
 T(\mathcal{C}) = \left\{0 \leq a \leq n-1 \mid \omega_2(\bar{a}) \leq \tfrac{m}{2} \right\},
\end{equation}
or
\begin{equation}\label{11defining set result 1}
 T(\mathcal{C}) = \left\{0 \leq a \leq n-1 \mid \omega_2(\bar{a}) \leq \tfrac{m}{2} \right\} \cup A_l.
\end{equation}
%
Then we have the following results.

\begin{lemma}\label{m even result 1}
Let the symbols be given as above.
Then $\mathcal{C}$ is a self-orthogonal binary cyclic code with parameters
$$\left[\,2^m-1,\; 2^{m-1} - \frac{1}{2}\binom{m}{\tfrac{m}{2}} - l m - 1,\; \geq 2^{\frac{m}{2}+1}\,\right],$$
where $l=0$ if $T(\mathcal{C})$ is given in (\ref{0318}).
\end{lemma}

\begin{proof}
We only prove the results for $T(\mathcal{C}) $ being given in (\ref{11defining set result 1}), and the proof of the results can be similarly obtained if
$T(\mathcal{C}) $ is given in (\ref{0318}), hence omitted here.

From the definition of $T(\mathcal{C})$ we immediately obtain
\begin{equation}\label{eq:112}
T(\mathcal{C})^{-1} = \left\{1 \le a \le n \;\middle|\; \omega_2(\bar{a}) \ge \frac{m}{2}\right\} \cup A_l^{-1},
\end{equation}
where $A_l^{-1} = \{n-i \mid i \in A_l\}$.
Since
\[
\mathbb{Z}_n = \left\{0 \le a \le n-1 \mid \omega_2(\bar{a}) \le \tfrac{m}{2}\right\}
          \cup \left\{0 \le a \le n-1 \mid \omega_2(\bar{a}) \ge \tfrac{m}{2}\right\},
\]
it follows that $\mathbb{Z}_n = T(\mathcal{C}) \cup T(\mathcal{C})^{-1}$.
By Lemma~\ref{criterion for self-orthogonal}, $\mathcal{C}$ is self-orthogonal.

We now determine the parameters of $\mathcal{C}$.
It is obvious that
\[
\Bigl|\bigl\{0 \le a \le n-1 \mid \omega_2(\bar{a}) \le \tfrac{m}{2}\bigr\}\Bigr|
   = \sum_{i=0}^{m/2} \binom{m}{i}.
\]
Using $2^m = \sum_{i=0}^{m} \binom{m}{i}$ and the symmetry of binomial coefficients,  we have then
\[
2^m = 2\sum_{i=0}^{m/2}\binom{m}{i} - \binom{m}{\tfrac{m}{2}}.
\]
Hence,
\begin{equation}\label{size of omega(a)}
\Bigl|\bigl\{0 \le a \le n-1 \mid \omega_2(\bar{a}) \le \tfrac{m}{2}\bigr\}\Bigr|
   = 2^{m-1} + \frac{1}{2}\binom{m}{\tfrac{m}{2}}.
\end{equation}

For each $i$ with $1 \le i \le l$, the $2$-adic expansion of $F_i$ is
\begin{equation}\label{eq:sd15}
\overline{F_i} = (\underbrace{1,1,\ldots,1}_{\frac{m}{2}},\underbrace{0,0,\ldots,0}_{f_i},1,\underbrace{0,\ldots,0}_{\frac{m}{2}-f_i-1}),
\end{equation}
where $1 \le f_i \le \frac{m}{2}-2$.
Then the $2$-adic expansion of $\operatorname{CL}(F_i)$ is
\begin{equation}\label{CL(I(t,m)) possible 1}
(\underbrace{1,1,\ldots,1}_{\frac{m}{2}},\underbrace{0,0,\ldots,0}_{f_i},1,\underbrace{0,0\ldots,0}_{\frac{m}{2}-f_i-1})
\end{equation}
if $f_i \le \bigl\lfloor\frac{m-2}{4}\bigr\rfloor$, or
\begin{equation}\label{CL(I(t,m)) possible 2}
(1,\underbrace{0,0,\ldots,0}_{{\frac{m}{2}-f_i-1}},\underbrace{1,1,\ldots,1}_{\frac{m}{2}},\underbrace{0,0,\ldots,0}_{f_i})
\end{equation}
if $f_i > \bigl\lfloor\frac{m-2}{4}\bigr\rfloor$.

If $F_i$ and $F_j$ were in the same cyclotomic coset, their coset leaders both could not be of type \eqref{CL(I(t,m)) possible 1} or type \eqref{CL(I(t,m)) possible 2}
since $1 \leq f_i\neq f_j\leq \frac{m}{2}- 2$.
If one leader were of type \eqref{CL(I(t,m)) possible 1} and the other were of type \eqref{CL(I(t,m)) possible 2}, then comparing the two forms would force $\frac{m}{2}=1$, which is impossible
for $m\ge6$.
Thus $F_i$ and $F_j$ lie in distinct cyclotomic cosets whenever $f_i \neq f_j$ with $1\le i\neq j\le l$.

For each $F_i$, we have $0 < -2^r(2^{f_i}-2^{\frac{m}{2}}+1) < n$ and
\[
2^{ \frac{m}{2} + r} \cdot F_i\equiv2^m-1+2^r(2^{f_i}-2^{\frac{m}{2}}+1) \pmod{n},
\]
where $r=0,1,\dots,\frac{m}{2}$. Clearly, for distinct $r_1,r_2 \in [0, \frac{m}{2}]$, $$2^m-1+2^{r_1}(2^{f_i}-2^{\frac{m}{2}}+1)\neq 2^m-1+2^{r_2}(2^{f_i}-2^{\frac{m}{2}}+1),$$ so the size of the cyclotomic coset $C_{F_i}$ satisfies $|C_{F_i}| \ge \frac{m}{2}+1$.
By Lemma~\ref{size about cyclotomic coset}, $|C_{F_i}|$ must divide $m$, hence $|C_{F_i}| = m$.
Moreover, from (\ref{eq:sd15}) every element $x \in C_{F_i}$ has $\omega_2(\bar{x}) = \frac{m}{2}+1$, so
\[
C_{F_i} \cap \left\{0 \le a \le n-1 \mid \omega_2(\bar{a}) \le \tfrac{m}{2}\right\} = \emptyset.
\]

Combining this with (\ref{size of omega(a)}) yields
\begin{equation}\label{m even 1 dim:df}
\begin{aligned}
\dim(\mathcal{C}) &= n - |T(\mathcal{C})| \\
&= n - \Bigl|\bigl\{0 \le a \le n-1 \mid \omega_2(\bar{a}) \le \tfrac{m}{2}\bigr\}\Bigr|
      - \sum_{i=1}^{l} |C_{F_i}| \\
&= 2^{m-1} - \frac{1}{2}\binom{m}{m/2} - l m - 1.
\end{aligned}
\end{equation}

Now consider any integer $u$ with $u \in [0, 2^{\frac{m}{2}+1}-2]$,
its $2$-adic expansion modulo $n$ is
\begin{equation}\label{eq:1130}
\overline{u} = (u_0,u_1,\ldots,u_{\frac{m}{2}},\underbrace{0,0,\ldots,0}_{\frac{m}{2}-1}),
\end{equation}
where $u_i \in \{0,1\}$ for $i=0,1,\dots,\frac{m}{2}$.
Because $2^{\frac{m}{2}+1}-1$ expands as
\[
(1,1,\ldots,1,\underbrace{0,0,\ldots,0}_{\frac{m}{2}-1}),
\]
at least one of the $u_i$ in (\ref{eq:1130}) is $0$. Hence $\omega_2(\bar{u}) \le \frac{m}{2}$, and consequently
\[
[0,\,2^{\frac{m}{2}+1}-2] \subseteq \left\{0 \le a \le n-1 \mid \omega_2(\bar{a}) \le \tfrac{m}{2}\right\}.
\]

Applying the BCH bound we obtain
\begin{equation}\label{eq:sdas18}
d(\mathcal{C}) \ge 2^{\frac{m}{2}+1}.
\end{equation}
From (\ref{m even 1 dim:df}) and (\ref{eq:sdas18}) the parameters of $\mathcal{C}$ are determined.
\end{proof}

\begin{lemma}\label{m even dual result 1}
Let the symbols be given as in Lemma \ref{m even result 1}.
Let $\mathcal{C}^{\perp}$ be the dual of $\mathcal{C}$.
Then $\mathcal{C}^{\perp}$ has parameters $\left[2^m-1,\; 2^{m-1} + \frac{1}{2}\binom{m}{m/2} + l m,  d(\C^\perp) \right]$, where
\begin{equation*}
\begin{split}
d(\mathcal{C}^{\perp})\geq \begin{cases}
2^{\frac{m}{2}}-1, &\text{if $l=0$},\\
2^{\frac{m}{2}} - 2^t - 1, &\text{if $l>0$,}
\end{cases}
\end{split}
\end{equation*}
and $t = \max\{f_1,f_2,\dots,f_l\}$.
\end{lemma}

\begin{proof} We proceed in a manner analogous to Lemma~\ref{m even result 1}.  We prove the conclusions  for the defining set $T(\mathcal{C})$ given by (\ref{11defining set
result 1}), the other case can be handled similarly and is omitted.

From Lemma~\ref{m even result 1},  the dimension of $\mathcal{C}^{\perp}$ is
\[
\dim(\mathcal{C}^{\perp}) = n - \dim(\mathcal{C}) = 2^{m-1} + \frac{1}{2}\binom{m}{m/2} + l m.
\]

In the following, our goal is to establish the lower bound
\[
d(\mathcal{C}^{\perp}) \ge 2^{\frac{m}{2}} - 2^{t} - 1.
\]
Using (\ref{eq:0131}) and (\ref{eq:112}),  the defining set of $\mathcal{C}^{\perp}$ is
\begin{equation}\label{eq:0205}
T(\mathcal{C}^{\perp}) = \bigl\{ 1 \leq a \leq n-1 \;\big|\; \omega_2(\bar{a}) \le \tfrac{m}{2}-1 \bigr\} \;\setminus\; A_{l}^{-1}.
\end{equation}
Consider any integer $v$ with $v \in [1, 2^{\frac{m}{2}}-2]$.
The $2$-adic expansion of $2^{\frac{m}{2}}-1$ modulo $n$ is
\[
(\underbrace{1,1,\ldots,1}_{\frac{m}{2}},\underbrace{0,0,\ldots,0}_{\frac{m}{2}}),
\]
then every $v$ in the range $[1,2^{\frac{m}{2}}-2]$ satisfies $\omega_2(\bar{v}) \le \tfrac{m}{2}-1$.  Therefore,
\begin{equation}\label{BCH bound with A is empty}
[1,\,2^{\frac{m}{2}}-2] \;\subseteq\; \bigl\{ 1\leq a \leq n-1 \mid \omega_2(\bar{a}) \le \tfrac{m}{2}-1 \bigr\}.
\end{equation}

From (\ref{eq:sd15}) we know that the $2$-adic expansion of $n-F_i$ modulo $n$ is
\[
\overline{n-F_{i}} = \bigl(\underbrace{0,0,\dots,0}_{\frac{m}{2}},\underbrace{1,1,\dots,1}_{f_i},0,\underbrace{1,1\dots,1}_{\frac{m}{2}-f_i-1}\bigr),
\]
where $1\le f_i\le \tfrac{m}{2}-2$. Consequently,
\[
\operatorname{CL}(n-F_{i}) = 2^{m/2} - 2^{f_i} - 1
= \bigl(\underbrace{1,1,\dots,1}_{f_i},0,\underbrace{1,1\dots,1}_{\frac{m}{2}-f_i-1},\underbrace{0,0,\dots,0}_{\frac{m}{2}}\bigr).
\]

Let $t = \max\{f_1,f_2,\dots,f_{l}\}$.  By the definition of $A_{l}^{-1}$ we have
\[
2^{m/2} - 2^{t} - 2 \notin A_{l}^{-1},\,\,\,\, \text{and} \,\,\,\,
2^{m/2} - 2^{t} - 1 \in A_{l}^{-1}.
\]

Combining the above facts with (\ref{eq:0205}) yields
\begin{equation}\label{BCH bound dual with result 1}
[1,\;2^{m/2} - 2^{t} - 2] \;\subseteq\; T(\mathcal{C}^{\perp}).
\end{equation}
Since $T(\mathcal{C}^{\perp})$ contains a consecutive set of $2^{m/2} - 2^{t} - 2$ nonzero integers, the BCH bound implies
\[
d(\mathcal{C}^{\perp}) \ge 2^{m/2} - 2^{t} - 1,
\]
which completes the proof.
\end{proof}

\begin{remark}\label{remark5}
If $l=0$, Lemmas \ref{m even result 1} and  \ref{m even dual result 1} hold when $m=4$. In this case, the code in Lemma \ref{m even result 1} has parameter $[15,4, \geq8]$ and the
code in Lemma \ref{m even dual result 1} has parameter $[15,11, \geq3]$. Both codes are optimal when equality holds, according to \cite{Ding2015,Grassl}.
\end{remark}

\begin{remark}\label{remark1}
Let $m\geq 6$, compared to the codes in Lemma \ref{lem:02061} and Lemma \ref{lem:0206}, the codes constructed in Lemma \ref{m even dual result 1} achieve a higher dimensions and a
 larger lower bound on the minimum distance under the same code length.
\end{remark}

\begin{example}
Let $m=6$ and $l=0$. The code in Lemma \ref{m even dual result 1} has parameter $[63,42, \geq7]$, the codes in Lemma \ref{lem:02061} and Lemma \ref{lem:0206} have parameters $[63,33,\geq5]$,
$[63,31,\geq3]$, $[63,30,\geq4]$, $[63,30,\geq6]$. In addition, the code $[64,42, 7]$ is optimal
according to \cite{Ding2015}.
\end{example}

\begin{example}
Let $m=8$ and $l=t=1$. The code in Lemma \ref{m even dual result 1} has parameter $[255,171, \geq13]$, the codes in Lemma \ref{lem:02061} and Lemma \ref{lem:0206} have parameters
$[255,127,\geq 9]$ and $[255,126,\geq10 ]$.
\end{example}

\begin{theorem}\label{result of section A}
Let the symbols be given as in Lemma \ref{m even result 1}. Let $h^{\perp}(x)$ denote the reciprocal polynomial of $(x^n+1)/g(x)$, where $g(x)$ is the generator polynomial of
$\mathcal{C}$.
Let $\mathcal{C}'$ denote the binary cyclic code of length $2n$ with generator polynomial $g(x)h^{\perp}(x)$. Then $\mathcal{C}'$ is a self-dual binary cyclic code with
parameters $[2n,n,\geq 2^{\frac{m}{2}+1}-2^{t+1}-2]$ or $[2n,n,\geq 2^{\frac{m}{2}+1}-2]$, where $t = \max\{f_1,f_2,\dots,f_l\}$. If $m=8$ and $t=1$, or $m\geq 10$, these lower bounds are larger than the square-root lower
bound.
\end{theorem}

\begin{proof}
By definition, $h^{\perp}(x)$ is the generator polynomial of $\mathcal{C}^\perp$. Let
$$\mathbb{C}=\{ (\mathbf{u} \mid \mathbf{u} + \mathbf{v}) : \mathbf{u} \in \mathcal{C}^\perp,\ \mathbf{v} \in \mathcal{C} \}.$$
By Lemma \ref{lemma:generalized-van-lint}, $\mathbb{C}$ is permutation-equivalent to $\mathcal{C}'$. The desired conclusions on the parameters of the code $\C'$ then follow from Lemma \ref{criterion for
plotkin-self-dual-2}, Lemma \ref{m even result 1} and Lemma \ref{m even dual result 1}.
It is easily seen that the bound $d(\C') \geq 2^{\frac{m}{2}+1}-2^{t+1}-2$ is larger than the square-root bound.
\end{proof}

\begin{example}\label{example1}
Let $m=8$. Then the self-dual cyclic code in Theorem \ref{result of section A} has parameter $[510,255, \geq30]$, and the minimum distance of the square-root lower bound is 23.
\end{example}

\begin{example}\label{example2}
Let $m=10$. Then the self-dual cyclic code in Theorem \ref{result of section A} has parameters $[2046,1023, \geq62]$, and the minimum distance of the square-root lower bound is 46.
\end{example}

\subsection{The second construction of self-dual binary cyclic codes}

In this subsection, let $n = 2^m - 1$, where $m$ is an even integer with $m \geq 6$. Let $l$ be a positive integer, $G_i= 2^{\frac{m}{2} -1} + 2^{ \frac{m}{2} +g_i-1}-1$ and
$B_l=C_{G_1}\cup C_{G_2}\cup \cdots \cup C_{G_l}$, where $1\leq i\leq l$ and $1 \leq g_1\neq g_2\neq\cdots\neq g_l\leq \frac{m}{2}- 2$. Then we have the following results.

\begin{lemma}\label{m even result 2}
Let the symbols be given as above and $\mathcal{C}$ be the binary cyclic code of length $n$ with defining set
\begin{equation}\label{defining set result 1}
 T(\mathcal{C}) = \left\{0 \leq a \leq n-1 \,\big|\, \omega_2(\bar{a}) \leq \tfrac{m}{2} \right\} \backslash B_l.
\end{equation}
Then $\mathcal{C}$ is a self-orthogonal binary cyclic code with parameters
$$\left[\,2^m-1,\; 2^{m-1} - \frac{1}{2}\binom{m}{m/2}+l m - 1,\; \geq 2^{\frac{m}{2}}+2^{\frac{m}{2}-1}\,\right].$$
\end{lemma}

\begin{proof}
It is clear that the 2-adic expansion of $G_i$ is
\begin{equation}\label{definte of Wi}
\overline{G_i}=(\underbrace{1,1,\ldots,1}_{\frac{m}{2}-1},\underbrace{0,0,\ldots,0}_{g_i},1,\underbrace{0,\ldots,0}_{\frac{m}{2}-g_i}).
\end{equation}
Then we obtain $\omega_2(\bar{a})=\frac{m}{2}$ for $a \in C_{G_i}$, which implies that $ B_l\subseteq \left\{0 \leq a \leq n-1 \,\big|\, \omega_2(\bar{a}) \leq \frac{m}{2} \right\} $.
Hence, with the same analysis as in Lemma \ref{m even result 1}, we obtain that  $\mathcal{C}$ is a self-orthogonal code.

We now show the parameters of $\mathcal{C}$.
By definition, it is easy to check that the 2-adic expansion of ${\rm CL}(G_i)$ is
\begin{equation}\label{CL(W(s,m)) possible 1}
\overline{{\rm CL}(G_i)}=(\underbrace{1,1,\ldots,1}_{\frac{m}{2}-1},\underbrace{0,0,\ldots,0}_{g_i},1,\underbrace{0,0,\ldots,0}_{ \frac{m}{2} - g_i})
\end{equation}
if $ g_i\leq \left\lfloor \frac{m}{4} \right\rfloor$, and
\begin{equation}\label{CL(W(s,m)) possible 2}
\overline{{\rm CL}(G_i)}=(1,\underbrace{0,0,\ldots,0}_{\frac{m}{2} - g_i},\underbrace{1,1,\ldots,1}_{\frac{m}{2}-1},\underbrace{0,0,\ldots,0}_{g_i})
\end{equation}
if $g_i >  \left\lfloor\frac{m}{4}\right\rfloor$.

If $G_i$ and $G_j$ were in the same cyclotomic coset, their coset leaders both could not be of type \eqref{CL(W(s,m)) possible 1} or type \eqref{CL(W(s,m)) possible 2}
since $1 \leq g_i\neq g_j\leq \frac{m}{2}- 2$.
If one leader were of type \eqref{CL(W(s,m)) possible 1} and the other were of type \eqref{CL(W(s,m)) possible 2}, then comparing the two forms would force $\frac{m}{2}=2$, which is impossible
for $m\ge6$.
Thus $G_i$ and $G_j$ lie in distinct cyclotomic cosets whenever $g_i \neq g_j$ with $1\le i\neq j\le l$.

From $1 \leq g_i \leq  \frac{m}{2}-1$, we have
$0<-2^r(2^{g_i}-2^{\frac{m}{2}}+1)<n,$
where $r=0,1,\ldots, \frac{m}{2}$.
Then
\begin{equation}\label{eq:ssd}
2^{ \frac{m}{2} + r} \cdot G_i\equiv2^m-1+2^r(2^{g_i}-2^{\frac{m}{2}}+1) \pmod{n},
\end{equation}
which means that  $|C_{G_i}| \geq  \frac{m}{2} +1$ for $1\leq i\leq l$. By Lemma~\ref{size about cyclotomic coset}, $|C_{G_i}|$ must be a divisor of $m$, hence $|C_{G_i}|=m$ for $1
\leq g_i \leq \frac{m}{2}- 1$.  Therefore, from (\ref{size of omega(a)}) we obtain
\begin{equation}\label{dim:df}
\begin{aligned}
\dim(\mathcal{C}) &= n - |T(\mathcal{C})| \\
&= n - \Bigl|\bigl\{0\leq a \leq n-1 \;\big|\; \omega_2(\bar{a}) \leq \tfrac{m}{2} \bigr\}\Bigr| + \sum_{i=1}^l |C_{G_i}| \\
&= 2^{m-1} - \frac{1}{2}\binom{m}{m/2} + l m - 1.
\end{aligned}
\end{equation}

From (\ref{CL(W(s,m)) possible 1}) and (\ref{CL(W(s,m)) possible 2}) we know that
\begin{equation}\label{CL(Wi)}
\overline{{\rm CL}(G_i)}\geq 2^{\frac{m}{2}}+2^{\frac{m}{2}-1}-1=(\underbrace{1,1,\ldots,1}_{\frac{m}{2}-1},0,1,\underbrace{0,0,\ldots,0}_{ \frac{m}{2}-1})
\end{equation}
for $1\leq i\leq l$.

Define $$V_0 = [0, 2^\frac{m}{2}-2]\; \text{and}\; V_1=[2^\frac{m}{2}, 2^\frac{m}{2}+2^{\frac{m}{2}-1}-2].$$
Let $ a \in V_0 $. It is easy to check that the $2$-adic expansion of $a$ must have the form
$$\overline{a}=(a_0,a_1,\ldots,a_{\frac{m}{2}-1},\underbrace{0,0,\ldots,0}_{ \frac{m}{2}-1}),$$
where $a_i \in \{0,1\}$ and at least one of $a_i=0$ for $i=0,1,\ldots,\frac{m}{2}-1$. Hence, $\omega_2(\bar{a})=\frac{m}{2}-1$ and $a<\overline{{\rm CL}(G_i)}$ for $1\leq i\leq l$,
which implies that
\begin{equation}\label{BCH bound dual with result 1}
V_0 \subseteq T(\mathcal{C}).
\end{equation}
Similarly, we can prove that $V_1 \subseteq T(\mathcal{C})$.
It is easily seen that $2^{\frac{m}{2}}-1 \in T(\C)$.
By the BCH bound, we have
\begin{equation}\label{BCH bound with resul 2 p}
d(\mathcal{C}) \geq  2^{\frac{m}{2}}+2^{\frac{m}{2}-1}-1.
\end{equation}

Because $0 \in T(\mathcal{C})$, the polynomial $x-1$ divides the generator polynomial of $\mathcal{C}$. Then $1+x+x^2+\cdots+x^{2^m-2} \in \mathcal{C}$.
This implies that every codeword of $\mathcal{C}$ has even Hamming weight.
Therefore, the bound in \ref{BCH bound with resul 2 p} can be strengthened to
\[
d(\mathcal{C}^{\perp}) \ge 2^{\frac{m}{2}} + 2^{\frac{m}{2}-1}.
\]
This completes the proof of this lemma.
\end{proof}

\begin{lemma}\label{m even dual result 2}
Let the symbols be given as in Lemma \ref{m even result 2}. Let ${\mathcal{C}}^{\perp}$ be the dual of $\mathcal{C}$, then $\mathcal{C}^{\perp}$ has parameters
$$\left[2^m-1,\; 2^{m-1}+\frac{1}{2}\binom{m}{\tfrac{m}{2}}-l m,\; \geq 2^{\tfrac{m}{2}}-1\right].$$

\end{lemma}

\begin{proof}
From Lemma~\ref{m even result 2}, the dimension of ${\mathcal{C}}^{\perp}$ is
\[
\dim(\mathcal{C}^{\perp}) = n - \dim(\mathcal{C}) = 2^{m-1} + \frac{1}{2}\binom{m}{\tfrac{m}{2}} - l m.
\]

Now consider any element \(a \in B_l^{-1}\), where $B_l^{-1} = \{n-i \mid i \in B_l\}$. By (\ref{definte of Wi}) and the definition of \(B_l\), we have \(\omega_2(\bar{a}) =
\frac{m}{2}\).

Using (\ref{eq:0131}) and (\ref{defining set result 1}), the defining set of \(\mathcal{C}^{\perp}\) can be expressed as
\[
T(\mathcal{C}^{\perp}) = \bigl\{1 \le a \le n-1 \mid \omega_2(\bar{a}) \le \tfrac{m}{2}-1 \bigr\} \cup B_l^{-1}.
\]

Then applying the BCH bound together with (\ref{BCH bound with A is empty}) yields the lower bound on the minimum distance
\[
d(\mathcal{C}^{\perp}) \ge 2^{\tfrac{m}{2}} - 1.
\]
This completes the proof.
\end{proof}

\begin{remark}\label{remark2}
Let $m\geq 8$, compared to the codes in Lemmas \ref{lem:02061} and \ref{lem:0206}, the codes constructed in Lemma \ref{m even dual result 2} achieve a higher dimensions and a
 larger lower bound on the minimum distance under the same code length.
\end{remark}

\begin{example}
Let $m=8$ and $l=1$. The code in Lemma \ref{m even dual result 2} has parameter $[255,155, \geq15]$, the codes in Lemma \ref{lem:02061} and Lemma \ref{lem:0206} have parameters $[255,127,\geq9]$
and $[255,126,\geq10]$, respectively.
\end{example}

\begin{example}
Let $m=10$ and $l=1$. The code in Lemma \ref{m even dual result 2} has parameter $[1023,628, \geq31]$, the codes in Lemma \ref{lem:02061} and Lemma \ref{lem:0206} have parameters $[1023,513,\geq17]$, $[1023,511,\geq9]$, $[1023,510,\geq10]$ and $[1023,512,\geq18]$.
\end{example}

Using an argument similar to the proof of Theorem \ref{result of section A}, Lemmas \ref{m even result 2} and \ref{m even dual result 2} yields the following result.

\begin{theorem}\label{result of section B}
Let the symbols be given as in Lemma \ref{m even result 2}.  Let $h^{\perp}(x)$ denote the reciprocal polynomial of $(x^n+1)/g(x)$, where $g(x)$ is the generator polynomial of
$\mathcal{C}$.
Let $\mathcal{C}'$ denote the binary cyclic code of length $2n$ with generator polynomial $g(x)h^{\perp}(x)$. Then $\mathcal{C}'$ is a self-dual binary cyclic code with
parameters $[2n,n,\geq 2^\frac{m}{2}+2^{\frac{m}{2}-1}]$. If $m \geq 8$, this minimum distance is larger than the square-root lower
bound.
\end{theorem}

\begin{example}\label{example4}
Let $m=8$. Then the self-dual cyclic code in Theorem \ref{result of section B} has parameters $[510,255, \geq 24]$, and the minimum distance of the square-root lower bound is 23.
\end{example}

\begin{example}\label{example3}
Let $m=10$. Then the self-dual cyclic code in Theorem \ref{result of section B} has parameters $[2046,1023, \geq 48]$, and the minimum distance of the square-root lower bound is 46.
\end{example}

\subsection{The third construction of self-dual binary cyclic codes}

In this subsection, let $n = 2^m - 1$, where $m$ is an odd integer. Let $l$ be a positive integer, $H_i= 2^{\frac{m-3}{2}+h_i} + 2^\frac{m-3}{2} - 1$ and $E_l=C_{H_1}\cup
C_{H_2}\cup \cdots \cup C_{H_l}$, where $1\leq i\leq l$ and $1 \leq h_1\neq h_2\neq\cdots\neq h_l\leq \frac{m-1}{2}$. Let
$\mathcal{C}$ be the binary cyclic code of length $n$ with the defining set
\begin{equation}\label{defining set result odd 1}
T(\mathcal{C})=\left\{1 \leq a\leq n-1 \,|\, \omega_2(\bar{a}) \leq \tfrac{m-3}{2} \right\} \cup C_{2^{(m-1)/2}-1},
\end{equation}
or
\begin{equation}\label{1defining set result odd 1}
T(\mathcal{C})=\left\{1 \leq a\leq n-1 \,|\, \omega_2(\bar{a}) \leq \tfrac{m-3}{2} \right\} \cup C_{2^{(m-1)/2}-1} \cup E_l.
\end{equation}
%
Then we have the following results.


%
%
%
%
%

\begin{lemma} \label{m odd result 1}
Let $m\geq 5$ and the symbols be given as above.
Then $\mathcal{C}$ is a dual-containing binary cyclic code with parameters
$$\left[2^m-1,\; 2^{m-1}+\binom{m}{\tfrac{m-1}{2}}-(l+1)m,\; \geq 2^{\tfrac{m-1}{2}}+2^{\tfrac{m-3}{2}}-1\right],$$
where $l=0$ if $T(\mathcal{C})$ is given in (\ref{defining set result odd 1}).
\end{lemma}

\begin{proof}
We prove the result only for the case that $T(\mathcal{C})$ is given by (\ref{1defining set result odd 1}), the proof for the other case is analogous and is therefore omitted.

It is straightforward to verify that the $2$-adic expansions of $2^{\frac{m-1}{2}}-1$ and $H_i$ are
\begin{equation}\label{eq:Hi expansion}
\overline{2^{\frac{m-1}{2}}-1}=(\underbrace{1,1,\ldots,1}_{\tfrac{m-1}{2}},0,0,\ldots,0)\,\,\,\, \text{and}\,\,\,\,
\overline{H_i}=(\underbrace{1,1,\ldots,1}_{\frac{m-3}{2}},\underbrace{0,0,\ldots,0}_{h_i},1,\underbrace{0,\ldots,0}_{\frac{m+1}{2}-h_i}),
\end{equation}
where $1\le i\le l$ and $1\le h_i\le\frac{m-1}{2}$, respectively.
Consequently, any element $a$ in $C_{2^{(m-1)/2}-1}\cup E_l$ satisfies $\omega_2(\bar{a})=\frac{m-1}{2}$, hence $\omega_2(\bar{a})\le\frac{m-1}{2}$ for all $a\in T(\mathcal{C})$.

From the definition of $T(\mathcal{C})$ we obtain
\begin{equation*}
T(\mathcal{C})^{-1}= \bigl\{1\le a\le n-1 \mid \omega_2(\bar{a})\ge \tfrac{m+3}{2}\bigr\}
                 \cup C_{2^{(m+1)/2}-1} \cup E_l^{-1},
\end{equation*}
where $E_l^{-1} = \{n-i \mid i \in E_l\}$.
A direct computation gives
\begin{equation}\label{m odd 1 T-2}
\overline{2^{\frac{m+1}{2}}-1}=(\underbrace{1,1,\ldots,1}_{\frac{m+1}{2}},0,0,\ldots,0)\,\,\,\, \text{and}\,\,\,\,
\overline{n-H_i}=(\underbrace{0,0,\ldots,0}_{\frac{m-3}{2}},\underbrace{1,1,\ldots,1}_{h_i},0,\underbrace{1,\ldots,1}_{\frac{m+1}{2}-h_i}),
\end{equation}
so that every $a\in T(\mathcal{C})^{-1}$ has $\omega_2(\bar{a})\ge\frac{m+1}{2}$.
Thus $T(\mathcal{C})\cap T(\mathcal{C})^{-1}=\emptyset$, and by Lemma~\ref{criterion for self-orthogonal} the code $\mathcal{C}$ is dual-containing.

We next compute the dimension of $\mathcal{C}$. It is clear that
\[
\Bigl|\bigl\{1\le a\le n-1 \mid \omega_2(\bar{a})\le \tfrac{m-3}{2}\bigr\}\Bigr|
   = \sum_{i=0}^{(m-3)/2}\binom{m}{i} - 1.
\]
Using $2^m=\sum_{i=0}^m\binom{m}{i}$ and the symmetry $\binom{m}{i}=\binom{m}{m-i}$,  we have
\[
2^m = 2\sum_{i=0}^{(m-3)/2}\binom{m}{i}+2\binom{m}{\tfrac{m-1}{2}},
\]
then
\begin{equation}\label{size of omega(a) odd}
\Bigl|\bigl\{1\le a\le n-1 \mid \omega_2(\bar{a})\le\tfrac{m-3}{2}\bigr\}\Bigr|
   = 2^{m-1}-\binom{m}{\tfrac{m+1}{2}}- 1.
\end{equation}

For each $H_i$, the $2$-adic expansion of its coset leader is
\begin{equation*}
(\underbrace{1,1,\ldots,1}_{\frac{m-3}{2}},\underbrace{0,0,\ldots,0}_{h_i},1,\underbrace{0,0\ldots,0}_{\frac{m+1}{2}-h_i})
\end{equation*}
if $h_i\le\bigl\lfloor\frac{m+1}{4}\bigr\rfloor$, and
\begin{equation*}
(1,\underbrace{0,0,\ldots,0}_{\frac{m+1}{2}-h_i},\underbrace{1,1,\ldots,1}_{\frac{m-3}{2}},\underbrace{0,0,\ldots,0}_{h_i})
\end{equation*}
if $h_i>\bigl\lfloor\frac{m+1}{4}\bigr\rfloor$.
Arguing as in Lemma \ref{m even result 2}, one shows that $H_i$ and $H_j$ lie in distinct cyclotomic cosets for $1 \leq i\neq j\leq l$.

For each $H_i$, we have $0 < -2^r(2^{h_i}-2^{\frac{m+3}{2}}+1) < n$ and
\[
2^{ \frac{m+3}{2} + r} \cdot H_i\equiv 2^m-1+2^r(2^{h_i}-2^{\frac{m+3}{2}}+1) \pmod{n},
\]
where $r=0,1,\dots,\frac{m-3}{2}$. Clearly, for distinct $r_1,r_2 \in [0, \frac{m-3}{2}]$,
 $$2^m-1+2^{r_1}(2^{h_i}-2^{\frac{m+3}{2}}+1)\neq 2^m-1+2^{r_2}(2^{h_i}-2^{\frac{m+3}{2}}+1),$$
 so the size of the cyclotomic coset $C_{H_i}$ satisfies $|C_{H_i}| \ge \frac{m-1}{2}$.

By Lemma~\ref{size about cyclotomic coset}, $|C_{H_i}|$ divides $m$, then $|C_{H_i}|=m$.
Moreover, from (\ref{eq:Hi expansion}) every $a\in C_{H_i}$ satisfies $\omega_2(\bar{a})=\frac{m-1}{2}$, and then
\[
E_l\cap\bigl\{1\le a\le n-1 \mid \omega_2(\bar{a})\le\tfrac{m-3}{2}\bigr\}=\emptyset.
\]
One also checks that $|C_{2^{(m+1)/2}-1}|=m$ and
\[
C_{2^{(m-1)/2}-1}\cap\bigl\{1\le a\le n-1 \mid \omega_2(\bar{a})\le\tfrac{m-3}{2}\bigr\}=\emptyset.
\]

Combining these facts with (\ref{size of omega(a) odd}) yields that
\begin{equation}\label{dim: m odd 1}
\begin{aligned}
\dim(\mathcal{C}) &= n-|T(\mathcal{C})| \\
&= n-\Bigl|\bigl\{1\le a\le n-1 \mid \omega_2(\bar{a})\le \tfrac{m-1}{2}\bigr\}\Bigr|
   -\sum_{i=1}^{l}|C_{H_i}| - |C_{2^{(m-1)/2}-1}| \\
&= 2^{m-1}+\binom{m}{\tfrac{m-1}{2}}-(l+1)m.
\end{aligned}
\end{equation}
Define two intervals
\[
N_0=\bigl[1,\;2^{\frac{m-1}{2}}-2\bigr]\,\,\,\, \text{and}\,\,\,\,
N_1=\bigl[2^{\frac{m-1}{2}},\;2^{\frac{m-1}{2}}+2^{\frac{m-3}{2}}-2\bigr].
\]
For $a\in N_0$, its $2$-adic expansion has the form $$(a_0,a_1,\dots,a_{\frac{m-3}{2}},0,\dots,0),$$
where $a_i \in \{0,1\}$ and at least one of $a_i=0$ for $i=0,1,\ldots,\frac{m-3}{2}$,  then $\omega_2(\bar{a})\le\frac{m-3}{2}$ and thus
$N_0\subseteq T(\mathcal{C})$.
For $a\in N_1$, the $2$-adic expansion is $$(a_0,a_1,\dots,a_{\frac{m-5}{2}},0,1,0,\dots,0),$$ where at least one $a_i=0$ for $i=0,1,\ldots,\frac{m-5}{2}$, which means that $\omega_2(\bar{a})\le\frac{m-3}{2}$. Thus,
$N_1\subseteq T(\mathcal{C})$.
Since $2^{\frac{m-1}{2}}-1\in C_{2^{(m-1)/2}-1}$, we have
\[
\bigl[1,\;2^{\frac{m-1}{2}}+2^{\frac{m-3}{2}}-2\bigr]\subseteq T(\mathcal{C}).
\]
Applying the BCH bound, we obtain
\[
d(\mathcal{C})\ge 2^{\frac{m-1}{2}}+2^{\frac{m-3}{2}}-1.
\]
Together with (\ref{dim: m odd 1}), the parameters of $\mathcal{C}$ are completely determined.
\end{proof}

\begin{remark}\label{remark3}
Let $m\geq 7$, compared to the codes in Lemmas \ref{lem:02061} and \ref{lem:0206}, the codes constructed in Lemma \ref{m odd result 1} achieve a higher dimensions and a larger lower bound on the minimum distance  under the same code length.
\end{remark}

\begin{example}\label{example7}
Let $m=5$ and $l=0$. The code in Lemma \ref{m odd result 1} has parameter $[31,21, \geq 5]$, and it is optimal when equality holds, according to \cite{Ding2015,Grassl}.
\end{example}

\begin{example}\label{example8}
Let $m=7$ and $l=0$. The code in Lemma \ref{m odd result 1} has parameters $[127,92,\geq 11]$, the codes in Lemma \ref{lem:02062} and Lemma \ref{lem:02063} have parameters $[127,64,\geq9]$
and $[127,63,\geq10]$, respectively. In addition, the code $[127,91,12]$ is optimal
according to \cite{Ding2015,Grassl}.
\end{example}

\begin{lemma} \label{m odd dual of result 2}
Let $m \geq 11$ be odd and the other symbols be given as in Lemma \ref{m odd result 1}. Let $\mathcal{C}^{\perp}$ be the dual of $\mathcal{C}$. Then code $\mathcal{C}^{\perp}$ has
parameters $$\left[2^m-1,\;2^{m-1}+(l + 1 )m - \binom{m}{\frac{m-1}{2}} - 1 ,\;\geq 2^{\frac{m+1}{2}}+2\right].$$
\end{lemma}
\begin{proof}
From Lemma~\ref{m odd result 1} we directly obtain the dimension of $\mathcal{C}^{\perp}$, which is given by
\[
\dim(\mathcal{C}^{\perp}) = n - \dim(\mathcal{C}) = 2^{m-1} + (l+1)m - \binom{m}{\frac{m-1}{2}}-1.
\]

By definition, the defining set of $\mathcal{C}^{\perp}$ is
\[
T(\mathcal{C}^{\perp}) = \left\{0 \le a \le n-1 \;\middle|\; \omega_2(\bar{a}) \le \tfrac{m+1}{2}\right\}
\;\setminus\;
\left( E_l^{-1} \;\cup\; C_{2^{(m+1)/2}-1} \right),
\]
where $E_l^{-1} = \{n-i \mid i \in E_l\}$.

Define
\[
V_0 = \bigl[0,\; 2^{\frac{m+1}{2}}-2\bigr]\,\,\,\, \text{and}\,\,\,\,
V_1 = \bigl[2^{\frac{m+5}{2}},\; 2^{\frac{m+5}{2}} + 2^{\frac{m+1}{2}}-2\bigr].
\]
We shall prove that both $V_0$ and $V_1$ are subsets of $T(\mathcal{C}^{\perp})$.

Take $a \in V_1$ and write $a = 2^{\frac{m+5}{2}} + b$ with $0 \le b \le 2^{\frac{m+1}{2}}-2$.
The $2$-adic expansion of $b$ has the form
\[
\bar{b} = (b_0, b_1, \dots, b_{\frac{m-1}{2}}, 0, \dots, 0),
\]
where $b_i \in \{0,1\}$ for $0\leq i\leq \frac{m-1}{2}$. Since
\[
\overline{2^{\frac{m+1}{2}}-1} = (\underbrace{1,1,\dots,1}_{\frac{m-1}{2}},0,0,\dots,0)
\]
and $b < 2^{\frac{m+1}{2}}-1$, at least one of the $b_i$ equals $0$.
Moreover,
\[
\overline{2^{\frac{m+5}{2}}} = (\underbrace{0,0,\dots,0}_{\frac{m+5}{2}},1,0,0,\dots,0),
\]
then the $2$-adic expansion of $a$ is
\[
\bar{a} = (b_0, b_1, \dots, b_{\frac{m-1}{2}}, 0,0,1,0,\dots,0).
\]
Therefore, $\omega_2(\bar{a}) \le \frac{m+1}{2}$.

If $\omega_2(\bar{a}) < \frac{m+1}{2}$, then by (\ref{m odd 1 T-2}) we have $a \notin E_l^{-1} \cup C_{2^{(m+1)/2}-1}$.
If $\omega_2(\bar{a}) = \frac{m+1}{2}$, then exactly one of $b_i$ is zero for $0\leq i\leq \frac{m-1}{2}$.
Consequently, $\bar{a}$ must be of the form
\begin{equation}\label{expansion of V}
\bar{a}=(\underbrace{1,1,\ldots,1}_{r},0,\underbrace{1,1,\ldots,1}_{\frac{m-1}{2}-r},0,0,1,\underbrace{0,\ldots,0}_{\frac{m-7}{2}}),
\end{equation}
with $0\leq r\leq \frac{m-1}{2}$.
Since $\overline{2^{\frac{m+1}{2}}-1}$ contains $\frac{m+1}{2}$ consecutive $1$'s, (\ref{expansion of V}) shows at most $\frac{m-1}{2}$ consecutive $1$'s. Hence, $a \notin
C_{2^{(m+1)/2}-1}$.
Furthermore, from (\ref{m odd 1 T-2}) the $2$-adic expansion of $n-H_i$ contains $\frac{m-3}{2}$ consecutive $0$'s, and by inspecting (\ref{expansion of V}) together with the definition of
$E_l$, one checks that $a \notin E_l^{-1}$ for $m\ge 11$.

Thus in all cases $a \in T(\mathcal{C}^{\perp})$, proving $V_1 \subseteq T(\mathcal{C}^{\perp})$.

The inclusion $V_0 \subseteq T(\mathcal{C}^{\perp})$ can be verified in an analogous manner, we omit the details. Consequently,
\begin{equation}\label{1V(s) belong defset}
V_0,\; V_1 \subseteq T(\mathcal{C}^{\perp}).
\end{equation}

Consider the set
\[
N = \left\{ b + i_1 c_1 + i_2 c_2 \mid 0\le i_1\le r_1,\; 0\le i_2\le r_2 \right\},
\]
with the choice
\[
b=0,\quad c_1=1,\quad c_2=2^{\frac{m+5}{2}},\quad r_1=2^{\frac{m+1}{2}}-2\,\,\,\, \text{and}\,\,\,\, r_2=1.
\]
Then $N = V_0 \cup V_1$, so by (\ref{1V(s) belong defset}) we have $N \subseteq T(\mathcal{C}^{\perp})$.
Clearly, $\gcd(n, c_1)=1$ and $\gcd(n, c_2)=1 < r_1+r_2+2$.
Applying Lemma~\ref{generalizations of bch bound.} yields
\begin{equation}\label{eq:ssdr}
d(\mathcal{C}^{\perp}) \ge 2 + r_1 + r_2 = 2^{\frac{m+1}{2}} + 1.
\end{equation}

Because $0 \in T(\mathcal{C}^{\perp})$, the polynomial $x-1$ divides the generator polynomial of $\mathcal{C}^{\perp}$. Then $1+x+x^2+\cdots+x^{2^m-2} \in \mathcal{C}$.
This implies that every codeword of $\mathcal{C}^{\perp}$ has even Hamming weight.
Therefore, the bound (\ref{eq:ssdr}) can be strengthened to
\[
d(\mathcal{C}^{\perp}) \ge 2^{\frac{m+1}{2}} + 2.
\]
This completes the proof of this lemma.
\end{proof}

Using an argument similar to the proof of Theorem~\ref{result of section A}, together with Lemmas \ref{m odd result 1} and \ref{m odd dual of result 2}, we obtain the following
result.

\begin{theorem}\label{result of section C}
Let $n=2^m-1$, where $m\geq 11$ is an odd. Let $h(x)$ denote the reciprocal polynomial of $(x^n+1)/g(x)$, where $g(x)$ is the generator polynomial of $C$.
Let $C'$ denote the binary cyclic code of length $2n$ with generator polynomial $g(x)h^{\perp}(x)$. Then $C'$ is a self-dual binary cyclic code with parameters $[2n,n,\geq
2^{\frac{m+1}{2}}+2]$. In addition, the bound $d(\C') \geq 2^{\frac{m+1}{2}}+2$ is larger than the square-root lower bound.
\end{theorem}

\begin{example}\label{example5}
Let $m=11$. It is clear that the self-dual cyclic code in Theorem \ref{result of section C} has parameter $[4094,2047, \geq66]$, and the minimum distance of the square-root lower bound is 64.
\end{example}

\subsection{The fourth construction of self-dual binary cyclic codes}

In this subsection, let $\C_{\delta}$ denote the narrow-sense primitive BCH code over $\mathbb{F}_2$ with length $n=2^m-1$ and designed distance
\begin{equation}\label{eq:0208}
2^{\frac{m+1}{2}}- 4s -5< \delta \leq 2^{\frac{m+1}{2}}- 4s-1,
\end{equation}
where $0\leq s \leq 2^\frac{m-5}{2}-1$ and $m\geq 5$ is odd. Let $\alpha$ be a primitive element in $\mathbb{F}_{2^m}$ and be a root of $\C_{\delta}$. By definition, we know that
the defining set of $\C_{\delta}$ with respect to $\alpha$ is  $T(\mathcal{C}_{\delta})= C_1\cup C_2 \cup \cdots \cup C_{\delta-1}$.

 The following two lemmas will be needed later.

\begin{lemma}\cite[Theorem 1]{Auly}\label{lem:1122}
Let $\C$ be a narrow-sense BCH code of length $n$ over $\mathbb{F}_q$ with designed distance $\delta$. Suppose that $m=ord_n(q)$.
Then $\C$ is dual-containing if the designed distance $\delta$ is in the range
\begin{equation*}
2\leq \delta \leq \frac{n}{q^m-1}(q^{\lceil m/2\rceil}-1-(q-2)\,[m \,\,odd]),
\end{equation*}
 where $[m \,\,odd]=1$ if $m$ is odd and $[m \,\,odd]=0$ if $m$ is even.
 \end{lemma}

\begin{lemma}\label{dimension of section D}\cite[Theorem~18]{Liu17}
Let  $m \geq 5 $ be an odd integer.  Put $h=\frac{m-1}{2}$ and $\delta_{N_q}=\delta-1-\lfloor \frac{\delta - 1}{q} \rfloor $. For $2 \leq \delta \leq  2^\frac{m+1}{2}+1$, we then have
$$\dim(\mathcal{C}_\delta)= n - m \delta_{N_q}. $$
\end{lemma}

We now develop a lower bound on the minimum distance of the dual code $\mathcal{C}^{\perp}_{\delta}$, which is better than the Carlitz-Uchiyama bound, the Sidel'nikov
bound, and the results in \cite{GDL21} and \cite{Wang24} in many cases.

\begin{theorem}\label{BCH code result1}
Let $[s \,\,odd]=1$ if $s$ is odd and $[s \,\,odd]=0$ if $s$ is even. Let $m \geq 5 $ be odd and $\C_{\delta}$ be the binary narrow-sense BCH code of length $2^m-1$ with designed
distance given in (\ref{eq:0208}). Then $\mathcal{C}_\delta^{\perp} $ is a self-orthogonal binary code with parameters
 \begin{equation}\label{eq:ss0208}
\begin{split}
 \begin{cases}
[n,\; m(2^\frac{m-1}{2}-2s-1),\; d(\mathcal{C}_\delta^{\perp})], \; & \text{if $2^{\frac{m+1}{2}}- 4s -3< \delta \leq 2^{\frac{m+1}{2}}- 4s-1$,}\\
[n,\; m(2^\frac{m-1}{2}-2s-2),\; d(\mathcal{C}_\delta^{\perp})], \; & \text{if $2^{\frac{m+1}{2}}- 4s -5< \delta \leq 2^{\frac{m+1}{2}}- 4s-3$,}
\end{cases}
\end{split}
\end{equation}
where $d(\mathcal{C}_\delta^{\perp}) \geq 2^\frac{m+1}{2}+s +[s \,\,odd]$ and $0\leq s \leq 2^\frac{m-5}{2}-1$.
\end{theorem}

\begin{proof}
Let $q=2$ and $n=2^m-1$ in Lemma \ref{lem:1122}. Then $\mathcal{C}_{\delta}$ is a dual-containing code, which means that $\mathcal{C}_{\delta}^{\perp}$ is self-orthogonal.

By the dimension formula in Lemma~\ref{dimension of section D}, the code $\mathcal{C}_{\delta}$ satisfies
\[
\dim(\mathcal{C}_{\delta}) =
\begin{cases}
m\bigl(2^{\frac{m-1}{2}}-2s-1\bigr), & \text{if } 2^{\frac{m+1}{2}}-4s-3 < \delta \le 2^{\frac{m+1}{2}}-4s-1,\\[4pt]
m\bigl(2^{\frac{m-1}{2}}-2s-2\bigr), & \text{if } 2^{\frac{m+1}{2}}-4s-5 < \delta \le 2^{\frac{m+1}{2}}-4s-3.
\end{cases}
\]
Hence, the dimension of its dual code is as given in (\ref{eq:ss0208}).

In the following, we prove the lower bound on the minimum distance of $\mathcal{C}_\delta^{\perp}$.
From the definition of $\mathcal{C}_{\delta}$ we have
\[
T(\mathcal{C}_{\delta}) = \bigcup_{i=1}^{\delta-1} C_i.
\]
It is easy to see that $\bigcup_{i=1}^{\delta-1} C_i^{-1} = \bigcup_{i=1}^{\delta-1} C_{n-i}$. Then by (\ref{eq:0131}) we have
\begin{equation*}
T(\mathcal{C}_{\delta}^{\perp}) = \mathbb{Z}_n \;\setminus\; \Bigl( \bigcup_{i=1}^{\delta-1} C_{n-i} \Bigr).
\end{equation*}
Define the intervals
\[
M_l = \bigl[\,2^{\frac{m+3}{2}}\cdot l,\; 2^{\frac{m+3}{2}}\cdot l + 2^{\frac{m+1}{2}}-2\,\bigr],
\]
where $0\le l\le s$. We shall prove that
\begin{equation}\label{V(s) belong defset}
\bigcup_{l=0}^{s} M_l \subseteq T(\mathcal{C}_{\delta}^{\perp}).
\end{equation}
There are two cases for discussion below.

\noindent {\bf Case 1:} $l=0$.
If $a\in M_0$, its $2$-adic expansion has the form
\begin{equation}\label{eq:0201}
\bar{a} = (a_0,a_1,\dots,a_{\frac{m-1}{2}},0,\dots,0).
\end{equation}
 Since $\overline{2^{\frac{m+1}{2}}-1}=(\underbrace{1,1,\dots,1}_{\frac{m+1}{2}},0,\dots,0)$ and $a<2^{\frac{m+1}{2}}-1$,
at least one of the $a_i$ is $0$.  Without loss of generality, assume $a_j = 0$ for $0 \leq j \leq \frac{m-1}{2}$. Then
\[
\overline{n-a} = (1-a_0,1-a_1,\dots,1-a_{\frac{m-1}{2}},1,1,\dots,1)
                \ge (\underbrace{0,0,\dots,0}_{j},1,0,\dots,0\underbrace{1,1,\dots,1}_{\frac{m-1}{2}}).
\]
Consequently,
\[
\operatorname{CL}(n-a) \ge (\underbrace{1,1,\dots,1}_{\frac{m+1}{2}},0,\dots,0)=2^{\frac{m+1}{2}}-1 \geq \delta.
\]
Thus $a\notin\bigcup_{i=1}^{\delta-1}C_{n-i}$, which means that $M_0\subseteq T(\mathcal{C}_{\delta}^{\perp})$.

\noindent {\bf Case 2:} $1\le l\le s$. When $m=5$, from the range of $s$, it can be seen that $l=0$. At this point, this situation does not exist. In this case, we always assume
that $m\geq 7$.

Let $c\in M_l$, it is clear that $c$ can be written uniquely as $c = a + b$ with $a\in M_0$ and $b = l\cdot 2^{\frac{m+3}{2}}$, we can write $l=
(b_0,b_1,\dots,b_{\frac{m-7}{2}},0,\dots,0)$ since $0\leq l \leq s \leq 2^\frac{m-5}{2}-1$.
Obviously,
the $2$-adic expansion of $b$ is
\begin{equation}\label{eq:020101}
\bar{b}=(\underbrace{0,0,\dots,0}_{\frac{m+3}{2}},b_0,b_1,\dots,b_{\frac{m-7}{2}},0).
\end{equation}
 Combining (\ref{eq:0201}) and (\ref{eq:020101}) we obtain
\[
\bar{c}= (a_0,a_1,\dots,a_{\frac{m-1}{2}},0,b_0,b_1,\dots,b_{\frac{m-7}{2}},0).
\]
Hence,
\begin{equation}\label{2adic expansion of n-b}
\overline{n-c}= (1-a_0,1-a_1,\dots,1-a_{\frac{m-1}{2}},1,1-b_0,1-b_1,\dots,1-b_{\frac{m-7}{2}},1).
\end{equation}

Suppose, for contradiction, that $c\in\bigcup_{i=1}^{\delta-1}C_{n-i}$. Then $n-c\in\bigcup_{i=1}^{\delta-1}C_i$, so there exists an integer $e$ with $1\le e\le\delta-1$ such that
$n-c\in C_e$. Because $\delta\le 2^{\frac{m+1}{2}}-1$, the $2$-adic expansion of $e$ is
\[
\bar{e}=(e_0,e_1,\dots,e_{\frac{m-1}{2}},0,\dots,0).
\]
If $n-c\in C_e$, comparing (\ref{2adic expansion of n-b}) with the form of $\bar{e}$ shows that the expansion of $n-c$ must contain a block of $\frac{m-1}{2}$ consecutive zeros.
This forces the vector $(a_0,a_1,\dots,a_{\frac{m-1}{2}})$ to be either
\[
(0,1,1,\ldots,1)\quad\text{or}\quad(1,1,1,\ldots,1,0)
\]
in (\ref{2adic expansion of n-b}). In the first case, a straightforward computation gives
\[
\overline{\operatorname{CL}(n-c)} = (1,1-b_0,1-b_1,\dots,1-b_{\frac{m-7}{2}},1,1,0,\dots,0)
                                 = 2^{\frac{m+1}{2}}-2l-1 \geq 2^{\frac{m+1}{2}}-2s-1.
\]
In the second case, we obtain
\[
\overline{\operatorname{CL}(n-c)} = (1,1,1-b_0,1-b_1,\dots,1-b_{\frac{m-7}{2}},1,0,\dots,0)
                                 = 2^{\frac{m+1}{2}}-4l-1 \geq 2^{\frac{m+1}{2}}-4s
                                 -1.
\]
However, any coset leader of a cyclotomic coset contained in $\bigcup_{i=1}^{\delta-1}C_i$ must satisfy
\[
1 \le \operatorname{CL}(\cdot) \le \delta-1 = 2^{\frac{m+1}{2}}-4s-2,
\]
which is less than $\operatorname{CL}(n-c)$. Hence, $n-c$ cannot belong to any $C_e$ with $e\le\delta-1$, a contradiction. Therefore, $c\notin\bigcup_{i=1}^{\delta-1}C_{n-i}$ and we have shown $M_l\subseteq
T(\mathcal{C}_{\delta}^{\perp})$ for every $l=1,\dots,s$.

Consider the set
\[
N = \Bigl\{ b + i_1c_1 + i_2c_2 + \dots + i_{s+1}c_{s+1}
      \;\Big|\; 0\le i_j\le r_j,\; 1\le j\le s+1 \Bigr\}
\]
with the choice
\[
b=0,\quad c_1=1,\quad c_2=c_3=\dots=c_{s+1}=2^{\frac{m+3}{2}},\quad
r_1=2^{\frac{m+1}{2}}-2,\quad r_2=r_3=\dots=r_{s+1}=1.
\]
Then $N = \bigcup_{l=0}^{s} M_l$, so by (\ref{V(s) belong defset}) we have $N\subseteq T(\mathcal{C}_{\delta}^{\perp})$.
Clearly, $gcd(n, c_j) = 1$ for $ 1 \leq j \leq s + 1 $, and this value is less than  $2 + r_1 + r_2 + \dots + r_{j-1}$ for $2 \leq j \leq s + 1 $.
Applying Lemma~\ref{generalizations of bch bound.} yields
\[
d(\mathcal{C}_{\delta}^{\perp}) \ge 2 + r_1 + r_2 + \dots + r_{s+1}
                                 = 2^{\frac{m+1}{2}} + s.
\]

Since $0\in T(\mathcal{C}_{\delta}^{\perp})$, the polynomial $x-1$ divides the generator polynomial of $\mathcal{C}_{\delta}^{\perp}$. Then $1+x+x^2+\cdots+x^{2^m-2} \in \mathcal{C}$.
This implies that every codeword of $\mathcal{C}_{\delta}^{\perp}$ has even Hamming weight.
If $s$ is odd, the lower bound $2^{\frac{m+1}{2}}+s$ is odd, so the actual minimum distance must be at least $2^{\frac{m+1}{2}}+s+1$.
Thus we obtain the refined bound
\[
d(\mathcal{C}_{\delta}^{\perp}) \ge
\begin{cases}
2^{\frac{m+1}{2}}+s, & \text{if $s$ is even},\\[2pt]
2^{\frac{m+1}{2}}+s+1, & \text{if $s$ is odd}.
\end{cases}
\]
This completes the proof.
\end{proof}

It is very hard to determine the minimum distance of $\mathcal C_{\delta}^\bot$ in general. Table 1 shows that the lower bounds in Theorem \ref{BCH code result1} are
 better than the Carlitz-Uchiyama bound, the Sidel'nikov
bound, and the results in \cite{GDL21} and \cite{Wang24} in many cases.

$${\rm Table\,\,\, 1:\,\,\,Lower\,\,\, bounds\,\,\, on\,\,\, the\,\,\, minimum\,\,\, distance\,\,\, of\,\,\, \mathcal{C}^{\perp}_{\delta}} \,\,\,{\rm for}\,\,n=2^m-1.$$

\resizebox{\textwidth}{!}{
\begin{tabular}{|c|c|c|c|c|c|c|}
\hline
$m$&$\delta$ & \multicolumn{5}{c|}{$d(\mathcal{\C_\delta}^{\perp}) \geq$} \\
\cline{3-7}
& & Sidel'nikvo bound~\cite{MS78} & Carlitz-Uchiyama bound~\cite{MS78}& Theorem 20~\cite{GDL21} & Theorem 1~\cite{Wang24} & Theorem \ref{BCH code result1}  \\
\hline
  &23 &16 &29 &31 &32 &34 \\
\cline{2-7}
9 &25 &16 &7 &31 &32 &34 \\
\cline{2-7}
  &27 &16 &$<0$ &31 &32 &34 \\
  \cline{2-7}
  &29 &16 &$<0$ &31 &32 &32 \\
\hline
  &47 &32 &28 &63 &64 &68\\
\cline{2-7}
11&49 &32 &$<0$ &63 &64 &68\\
\cline{2-7}
  &51 &32 &$<0$ &63 &64 &66\\
\cline{2-7}
  &53 &32 &$<$0 &63 &64 &66\\
\hline
  &91 &64 &113 &127 &128 &136\\
\cline{2-7}
13&93 &64 &23 &127 &128 &136\\
\cline{2-7}
  &103 &64 &$<0$ &127 &128 &134\\
\cline{2-7}
  &105 &64 &$<0$ &127 &128 &134\\
\hline
\end{tabular}
}

\vspace{8mm}

From the BCH bound, we know that $d(\mathcal{C}_\delta)\geq \delta$. Then using an argument similar to the proof of Theorem~\ref{result of section A}, together with Theorem~\ref{m odd result 1}, we obtain the following result.
\begin{theorem}\label{result of section D}
Let the symbols be given as in Theorem \ref{BCH code result1} and $2^{\frac{m+1}{2}}- 4s -5< \delta \leq 2^{\frac{m+1}{2}}- 4s-1$ for $0\leq s \leq 2^\frac{m-5}{2}-1$. Let $h(x)$ denote the reciprocal polynomial of $(x+1)^n/g(x)$, where $g(x)$ is the generator polynomial of $\mathcal{C}_\delta$. Let $\mathcal{C}'$ denote the binary cyclic code of length $2n$ with generator polynomial $g(x)h^\perp(x)$. Then $\mathcal{C}'$ is a self-dual binary cyclic code with parameters $[2n,n,d(\mathcal{C}')]$, where
\[
d(\mathcal{C}') \geq \min\Bigl\{2\delta,\; 2^{\frac{m+1}{2}}+s+[s\text{ odd}]\Bigr\},
\]
 and this lower bound is larger than the square-root lower bound for $\delta\neq 2^{\frac{m-1}{2}}$.
\end{theorem}

\begin{example}\label{example6}
Let $m=7$ and $s=1$. Then the self-dual cyclic code in Theorem~\ref{result of section D} has parameter [254,127,18], the minimum distance of the square-root lower bound is 16.
\end{example}

\begin{example}\label{example9}
Let $m=9$ and $s=3$. Then the self-dual cyclic code in Theorem~\ref{result of section D} has parameter [1022,511,36], the minimum distance of the square-root lower bound is 32.
\end{example}

\section{Summary and concluding remarks}
The main contributions of this paper are the following:
\begin{itemize}
\item We presented four methods for constructing infinite families of self-dual binary cyclic codes of length $2^{m+1}-2$, where
the minimum distances of the constructed codes surpass the square-root lower bound (see Theorems \ref{result of section A}, \ref{result of section B}, \ref{result of section C} and
\ref{result of section D}, and Examples \ref{example1}, \ref{example2}, \ref{example3}, \ref{example4}, \ref{example5},  \ref{example6},  \ref{example9}). Consequently, we provided an affirmative answer to the
70-year-old open problem in coding theory (see Open Problem 1).

\item We constructed several classes of binary cyclic codes with larger dimensions and larger lower bounds  on the minimum distances (see Lemmas \ref{m even dual result 1}, \ref{m even dual result 2},  and
\ref{m odd result 1} and Theorem \ref{BCH code result1}),  compared with the codes constructed in \cite{TangDing22}
 (see Remarks \ref{remark1}, \ref{remark2} and \ref{remark3}).  Some of these binary cyclic codes are distance-optimal (see Remark  \ref{remark5}, and Examples \ref{example7}, \ref{example8}).

\item We gave a lower bound on the minimum distance of the dual code $\mathcal{C}^{\perp}_{\delta}$ for $2^{\frac{m+1}{2}}- 4s -5< \delta \leq 2^{\frac{m+1}{2}}- 4s-1$ and $0\leq s \leq 2^\frac{m-5}{2}-1$, which is better than the Carlitz-Uchiyama bound, the Sidel'nikov
bound, and the results in \cite{GDL21} and \cite{Wang24} in many cases (see  Theorem \ref{BCH code result1} and Table 1).
\end{itemize}

The main technique for constructing self-dual binary cyclic codes used in this paper is  Lemma
\ref{criterion for plotkin-self-dual-1}.  To construct a $[2n,n,d]$  self-dual binary cyclic code such that
$d/2n$ is as large as possible using
this technique,   one must design the building block (i.e.,  the cyclic code $\C$) carefully.  The key
point is to construct a binary cyclic code $\C$ (the building block of the technique) such that
\begin{itemize}
\item $\C^\perp \subseteq \C$ and $\min\{ d(\C^\perp), 2d(\C)\}$ is as large as possible; or
\item $\C \subseteq \C^\perp$ and $\min\{ d(\C), 2d(\C^\perp)\}$ is as large as possible.
\end{itemize}
This explains the difficulty of using this technique.  It would be an interesting and challenging problem to construct an infinite family of self-dual binary cyclic codes of length $2^{m+1}-2$ having a lower bound on the minimum distances better than those presented in this paper.

The reader is informed that infinite families of $[n, k,d]$ binary
cyclic codes  with $k \geq (n-1)/2$ and a lower bound on $d$ much better than the square-root bound were constructed in \cite{Chen26} and
\cite{SLD}, but those binary codes are not self-dual.

\end{document}